\documentclass[preprint,12pt]{elsarticle}
\usepackage{amsmath}
\usepackage{amssymb}
\usepackage{tabularx}
\usepackage{graphics}
\usepackage{epsfig}
\usepackage{vaucanson-g}
\usepackage[whileod]{chl-alg}

\newtheorem{definition}{\textbf{Definition}}
\newproof{proof}{Proof}
\newdefinition{rmk}{Remark}
\newtheorem{thm}{Theorem}
\newtheorem{lemma}{Lemma}
\newtheorem{prop}{Proposition}
\def\dd{\mathinner{.\,.}}

\journal{Theoretical Computer Science}
\begin{document}

\begin{frontmatter}

\title{Linear-Time Superbubble Identification Algorithm for Genome Assembly}

\author[label1,label2]{Ljiljana Brankovic}
\ead{ljiljana.brankovic@newcastle.edu.au}

\author[label2]{Costas~S.~Iliopoulos}
\ead{costas.iliopoulos@kcl.ac.uk}
\author[label2]{Ritu Kundu}
\ead{ritu.kundu@kcl.ac.uk}
\author[label2]{Manal Mohamed}
\ead{manal.mohamed@kcl.ac.uk}
\author[label2]{Solon~P.~Pissis\corref{cor1}}
\ead{solon.pissis@kcl.ac.uk}
\author[label2]{Fatima Vayani}
\ead{fatima.vayani@kcl.ac.uk}

\cortext[cor1]{Corresponding author}   

\address[label1]{School of Electrical Engineering and Computer Science, The University of Newcastle,\\ Newcastle NSW 2308, Australia.}
  
\address[label2]{Department of Informatics, King's College London,\\ London WC2R 2LS, United Kingdom}

\begin{abstract}
DNA sequencing is the process of determining the exact order of the nucleotide bases of an individual's genome in order to catalogue sequence variation and understand 
its biological implications.
Whole-genome sequencing techniques produce masses of data in the form of short sequences known as reads. 
Assembling these reads into a whole genome constitutes a major algorithmic challenge.
Most assembly algorithms utilise de Bruijn graphs constructed from reads for this purpose.
A critical step of these algorithms is to detect typical motif structures in the graph caused by sequencing errors and genome repeats, and filter them out;
one such complex subgraph class is a so-called {\em superbubble}.
In this paper, we propose an $\mathcal{O}(n+m)$-time algorithm to detect all superbubbles in a directed acyclic graph with $n$ vertices and $m$ (directed) edges, improving the best-known $\mathcal{O}(m \log m)$-time algorithm by Sung et al.
\end{abstract}

\end{frontmatter}

\section{Introduction}
\label{sec-intro}

Since the publication of the first draft of the human genome~\cite{lander2001,venter2001}, the field of genomics has changed dramatically. Recent developments in sequencing technologies (see~\cite{Solexa}, for example) have made it possible to sequence new genomes at a fraction of the time and cost required only a few years ago. With applications including sequencing the genome of a new species, an individual within a population,
and 
RNA molecules from a particular sample, 
sequencing remains at the core of genomics. 

Whole-genome sequencing creates masses of data, in the order of tens of gigabytes, in the form of short sequences (reads). Genome assembly involves piecing together these reads 
to form a set of contiguous sequences (contigs) representing the DNA sequence in the sample. Traditional assembly algorithms rely on the overlap-layout-consensus approach~\cite{batzoglou2005}, representing each read as a vertex in an overlap graph and each detected overlap as a directed edge between the vertices corresponding to overlapping reads. These methods have proved their use through numerous {\it de novo} genome assemblies~\cite{butler2008}.

Subsequently, a fundamentally different approach based on de Bruijn graphs was proposed~\cite{Eulerian}, where representation of data elements was organised around words of $k$ nucleotides, or $k$-mers, instead of reads. Unlike in an overlap graph, in a {\em de Bruijn graph}~\cite{deBruijn}, each $k-1$ nucleotide long prefix and suffix of the $k$-mers is represented as a vertex and each $k$-mer is represented as a directed edge between its prefix and suffix vertices. The marginal information contained by a $k$-mer is its last nucleotide. The sequence of those final nucleotides is called the sequence of the vertex.  In a de Bruijn graph, the assembly problem is reduced to finding  an  Eulerian path, that is, a trail  that visits each edge in the graph exactly once. 

However, sequencing errors and genome repeats significantly complicate the de Bruijn graph by adding false vertices and edges to it.  Efficient and robust filtering methods have been proposed to simplify the graph by filtering out motifs such as tips, bubbles, and cross links, which proved to be caused by sequencing errors~\cite{velvet}. In particular, a {\em bubble} consists of multiple directed unipaths (where a unipath is a path in which all internal vertices are of degree $2$) from a vertex $v$ to a vertex $u$ and is commonly caused by a small number of errors in the centre of  reads. Although these types of motifs are simple and can easily be identified and filtered out, more complicated motifs prove to be more challenging.

Recently, a complex generalisation of a bubble, the so-called superbubble, was proposed as an important subgraph class for analysing assembly graphs~\cite{OnoderaSS13}. A \emph{superbubble} is defined as a minimal subgraph $H$ in the de Bruijn graph with exactly one start vertex $s$ and one end vertex $t$ such that: (1) $H$ is a directed, acyclic, single-source ($s$), single-sink ($t$) graph (2) there is no edge from a vertex not in $H$ going to a vertex in $H \backslash \{s\}$ and (3) there is no edge from a vertex in $H \backslash \{t\}$ going to a vertex not in $H$. It is clear that many  superbubbles are formed as a result of sequencing errors, inexact repeats, diploid/polyploid genomes, or frequent mutations. Thus, efficient detection of superbubbles is essential for the application of genome assembly~\cite{OnoderaSS13}.

Onodera et al.~gave an $\mathcal{O}(nm)$-time algorithm to detect superbubbles,  where $n$ is the number of vertices and $m$ is the number of edges in the graph~\cite{OnoderaSS13}. Very recently, Sung et al.~gave an improved $\mathcal{O}(m \log m)$-time algorithm to solve this problem~\cite{SungSSBP14}.  The  algorithm partitions  the given graph into a set of subgraphs such that the set of superbubbles in all these subgraphs is the same as the set of superbubbles in the given graph. This set consists of subgraphs corresponding to each non-singleton strongly  connected component and a subgraph corresponding to the set of all the vertices involved in singleton strongly connected components.  Superbubbles are then detected in each subgraph; if it is cyclic, it is first converted into a directed acyclic subgraph by  means of depth-first search and  by duplicating some vertices.  

\textbf{Our Contribution}. Note that the cost of partitioning the graph and transforming it into the directed acyclic subgraphs is linear with respect to the size of the graph. However, computing the superbubbles in each directed acyclic subgraph requires $\mathcal{O}(m \log m)$ time~\cite{SungSSBP14}, which dominates the time bound of the algorithm. In this paper, we propose a new $\mathcal{O}(n+m)$-time algorithm to compute all superbubbles in a directed acyclic graph.

This paper is organised as follows: In Section 2 we define superbubbles and introduce some of their properties,  and in  Section 3 we outline the $\mathcal{O}(n+m)$-time algorithm for computing superbubbles in a directed acyclic graph. In Section 4 we explain a method to validate a candidate  superbubble in constant time. The algorithm is analysed in Section 5, while Section 6 provides some final remarks and directions for future research. 

\section{Properties}
\label{sec-superbubbles}

\noindent The concept of  superbubbles was introduced and formally defined in \cite{OnoderaSS13} as follows. 
\begin{definition} [\cite{OnoderaSS13}]\label{def:sup}
Let $G = (V, E)$ be a directed graph. For any  ordered pair of
distinct vertices $s$ and $t$,  $\langle s, t\rangle$ is called a {\sl superbubble} if it satisfies the following:
\begin{itemize}
\item {\bf reachability:}  $t$ is reachable from s;
\item {\bf matching:} the set of vertices reachable from $s$ without passing
through $t$ is equal to the set of vertices from which $t$ is reachable without passing
through $s$;
\item {\bf acyclicity:} the subgraph induced by $U$ is acyclic,  where $U$ is the set of
vertices satisfying the matching criterion;
\item {\bf minimality:} no vertex in $U$ other than $t$ forms a pair with $s$ that satisfies
the conditions above;
\end{itemize}
%
%
vertices $s$ and $t$, and $U \backslash \{s, t\}$ used in the above definition are the superbubble's {\sl entrance},
{\sl exit} and {\sl interior}, respectively. 

\end{definition}

We note that  a superbubble $\langle s, t\rangle$ in the above definition  is equivalent to a single-source, single-sink, directed acyclic subgraph  of $G$ with source $s$ and sink $t$, which does not have any cut vertices and preserves all in-degrees and out-degrees of vertices in $U\backslash \{s, t\}$, as well as the out-degree of $s$ and in-degree of $t$.

We next state a few important properties of superbubbles which enable the linear-time enumeration of superbubbles.  Lemmas~\ref{entrance} and~\ref{exit} were proved by Onodera et al.~\cite{OnoderaSS13} and Sung et al.~\cite{SungSSBP14}, respectively.

\begin{lemma} [\cite{OnoderaSS13}]\label{entrance}
Any vertex can be the entrance (respectively exit) of at most
one superbubble.
\end{lemma}

Note that Lemma~\ref{entrance} does not exclude the possibility that a vertex
is the entrance of a superbubble and the exit of another superbubble. 

\begin{lemma}[\cite{SungSSBP14}] \label{exit}
Let $G$ be a directed acyclic graph. We have the following two observations.

1) Suppose $(p, c)$ is an edge in $G$, where $p$ has one child and $c$ has one parent, then $\langle p, c \rangle$ is a superbubble in $G$.

2) For any superbubble $\langle s, t \rangle$ in $G$, there must exist some parent $p$ of $t$ such that $p$ has exactly one child $t$.
\end{lemma}

In this paper we start by showing another important property of superbubbles that is closely-related to Lemma~\ref{exit}.

\begin{lemma}\label{start}
For any superbubble $\langle s, t \rangle$ in a directed acyclic graph $G$, there must exist some child $c$ of $s$ such that $c$ has exactly one parent $s$.
\end{lemma}

\begin{proof}
Assume that all the children of $s$ have more than one parent. Then, there must be some cycle or some child $c$ which has a parent that does not belong to the superbubble $\langle s, t \rangle$. This is a contradiction. \qed
\end{proof}

\section{Finding a Superbubble in a Directed Acyclic Graph}\label{sec:find}
 
The main contribution of this paper is an algorithm \Algo{SuperBubble} that reports  {\em all}  superbubbles in a directed acyclic graph $G=(V,E)$ with exactly  one source (the vertex with in-degree 0) and exactly one sink  (vertex with out-degree 0). If $G$ has more than one source then a new source vertex  $r'$ is added to $V$ and an edge from $r'$ to each  existing source  is added to $E$. The same is done if $G$ has more than one sink; in this case, a new sink vertex $t'$ is added to $V$ and an edge from each existing sink to $t'$ is  added to $E$. If such preprocessing is done, then among the superbubbles reported by the algorithm, only those which do not start at $r'$ and do not end at $t'$ represent the superbubbles in the original graph. For the sake of simplicity, for the rest of this paper and in all the propositions, lemmas and theorems that follow, we use $G$ to denote a directed acyclic graph  with exactly one source and exactly one sink, and we use $n$ and $m$ to denote the number of its vertices and edges respectively, that is, for $G=(V,E)$ we have  $n= |V|$ and $m=|E|$.


A {\em topological ordering} \textit{ordD} of $G$ maps each vertex to an integer between $1$ and $n$, such that $\textit{ordD}[x] < \textit{ordD}[y]$ holds for all edges $(x,y) \in E$. There exists a classical linear-time algorithm for computing the topological ordering of a directed acyclic graph~\cite{Cormen2001,Tarjan67}.  
In its recursive form, the algorithm visits an unvisited vertex of the graph, finds its unvisited neighbour, say $v$, and performs another topological sort starting from $v$. The algorithm {\em returns} if the current vertex does not have unvisited neighbours. Algorithm \Algo{TopologicalSort}, given below, is a simplified version that takes as input a single-source, single-sink directed acyclic graph, and produces a topological ordering of vertices. For the  graph $G$ in Figure~\ref{figu-graph}, \Algo{TopologicalSort} produces an ordering given in Figure~\ref{figu-topological}.
%
%

\bigskip\noindent

\begin{algo}{TopologicalSort}{G}
\SET{\textit{order}}{n}
\DOFOR{\mbox{each vertex $v \in V$}}
\SET{state[v]}{unvisited}
\OD

\CALL{RecursiveTopologicalSort}{G,\textit{source}}

\end{algo}
\bigskip

\begin{algo}{RecursiveTopologicalSort}{G, v}
\SET{state[v]}{visited}
\DOFOR{\mbox{each vertex $w$ adjacent to $v$}}
\IF{ state[w] = unvisited}
  \CALL{RecursiveTopologicalSort}{G,w}
\FI
\OD
\SET{ordD[v]}{order}
\SET{order}{order-1}
\end{algo}




\bigskip

\begin{prop} \label{ts1}
For any topological ordering $\textit{ordD}$ of vertices in graph $G$, if  vertex $u$ is reachable from $v$, that is, if there is a path from $v$ to $u$, then \textit{ordD}$[v] <$ \textit{ordD}$[u]$.
\end{prop}

\begin{proof}
If the path from $v$ to $u$ is of length $1$, i.e.,~there is an edge $(v,u)$, then by the definition of topological ordering we have \textit{ordD}$[v] < $ \textit{ordD}$[u]$. Otherwise, we denote the path   from $v$ to $u$  of length $k$,  $k>1$,  as $v,x_1,\ldots,x_{k-1},u$. Then by the definition of topological ordering we have \textit{ordD}$[v] < $ \textit{ordD}$[x_1] < \cdots <$  \textit{ordD}$[u]$. Transitively, we have \textit{ordD}$[v] < $ \textit{ordD}$[u]$.
\qed

\end{proof}
Importantly, in this paper we do not consider all topological orderings of graph $G$ but only those obtained by algorithm \Algo{TopologicalSort}. Note that this algorithm finds a  directed spanning tree $T$ of $G$ rooted at the \textit{source}, which contains a path from the \textit{source} to any vertex in $G$. The directed spanning tree $T$ of $G$  obtained by algorithm \Algo{TopologicalSort} is presented by bold edges in Figure~\ref{figu-topological}. It may be worth mentioning that a directed rooted tree is also known as \textit{arborescence}.   

We next present another important property of topological ordering obtained by algorithm \Algo{TopologicalSort}.

\begin{prop}\label{ts2}
Let \textit{ordD} and $T$ be a topological ordering and a directed rooted spanning tree of graph $G$ obtained by algorithm \Algo{TopologicalSort}. If  there is a path in $T$  from a vertex  $v$ to a vertex $u$, then, for each vertex $w$ such that   $\textit{ordD}[v]<  \textit{ordD}[w] < \textit{ordD}[u]$, there is a path from $v$ to $w$.
\end{prop}

\begin{proof}
Recall that  $T$ contains a path from the root to each  vertex of the tree; this is also true for each subtree of $T$.  Furthermore, if there is a path from $v$ to $u$ in $T$, then $u$ is contained in a subtree of $T$ rooted at $v$, and each $w$ such that $\textit{ordD}[v]<  \textit{ordD}[w] < \textit{ordD}[u]$ is also contained in the subtree rooted at $v$ (but not in the subtree rooted at $u$). Therefore, there is a path from $v$ to $w$, for each $w$  such that $\textit{ordD}[v]<  \textit{ordD}[w] < \textit{ordD}[u]$.
\qed 
\end{proof}

%
%

We next show that in an ordering obtained by  \Algo{TopologicalSort}, a vertex has the topological ordering between the orderings of the entrance and the exit of a superbubble if and only if it belongs to the superbubble. 

\begin{lemma}\label{ts}
Let  graph $G$ contain a superbubble $\langle s, t\rangle$. Then a topological ordering obtained by  \Algo{TopologicalSort} has the following properties.
\begin{enumerate}
\item For all $x$ such that $x \in U \backslash \{s, t\}$, $\textit{ordD}[s] < \textit{ordD}[x] < \textit{ordD}[t]$.
\item  For all $y$ such that $y \not\in U$,  $\textit{ordD}[y]< \textit{ordD}[s]$ or $\textit{ordD}[y] > \textit{ordD}[t]$. 
\end{enumerate}
\end{lemma}

\begin{proof}
Recall that $U$ is the set of vertices forming a superbubble  (see Definition~\ref{def:sup}).
\begin{enumerate}
\item Since there is a path from the start $s$ of the superbubble to all $x\in U \backslash \{s\}$, by Proposition~\ref{ts1} we have $\textit{ordD}[s] < \textit{ordD}[x]$  for all $x$ such that $x \in U \backslash \{s\}$. 
Similarly,  since there is a path from all $x\in U \backslash \{t\}$ to the end $t$ of the superbubble, by Proposition~\ref{ts1} we have $\textit{ordD}[x] < \textit{ordD}[t]$    for all $x$ such that $x \in U \backslash \{t\}$. Therefore, for all $x$ such that $x \in U \backslash \{s, t\}$, $\textit{ordD}[s] < \textit{ordD}[x] < \textit{ordD}[t]$.

\item  Suppose the opposite, that is, suppose that there exists some $y \not\in U$ such that $\textit{ordD}[s] < \textit{ordD}[y] < \textit{ordD}[t]$. Since the superbubble $\langle s, t\rangle$ is itself a single-source, single-sink subgraph of $G$, any directed spanning  tree of $G$ rooted at the \textit{source}, will contain a path from $s$ to $t$. Then by Proposition~\ref{ts2} there also exists a path from $s$ to $y$ in $T$ and thus also in $G$. However, by the definition of the superbubble, the only vertices reachable from $s$ without going through $t$ are the internal vertices of the superbubble --- a contradiction. Therefore, for all $y$ such that $y \not\in U$, either $\textit{ordD}[y]< \textit{ordD}[s]$ or $\textit{ordD}[y] > \textit{ordD}[t]$. 
\qed
\end{enumerate}
\end{proof}

\begin{figure}[!t]
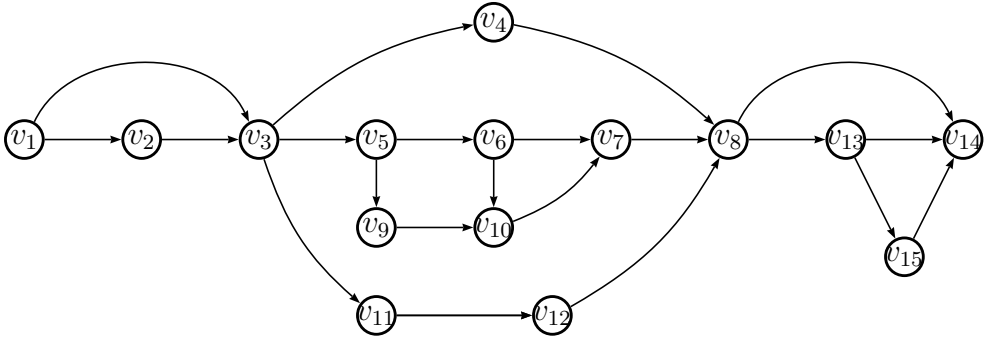

\begin{center}
\VCDraw[1.3]{
\begin{VCPicture}{(18,2)(0,-3)}
\State[v_1]{(0,0)}{1}
\State[v_2]{(2,0)}{2}
\State[v_3]{(4,0)}{3}
\State[v_5]{(6,0)}{5}
\State[v_6]{(8,0)}{6}
\State[v_7]{(10,0)}{7}
\State[v_8]{(12,0)}{8}
\State[v_4]{(8,2)}{4}
\State[v_9]{(6,-1.5)}{9}
\State[v_{10}]{(8,-1.5)}{10}
\State[v_{11}]{(6,-3)}{11}
\State[v_{12}]{(9,-3)}{12}

\State[v_{13}]{(14,0)}{13}
\State[v_{14}]{(16,0)}{14}
\State[v_{15}]{(15,-2)}{15}

\EdgeL{1}{2}{}
\EdgeL{2}{3}{}
\EdgeL{3}{5}{}
\EdgeL{5}{6}{}
\EdgeL{6}{7}{}
\EdgeL{7}{8}{}
\EdgeL{9}{10}{}
\EdgeL{11}{12}{}
\EdgeL{8}{13}{}
\EdgeL{13}{14}{}
\EdgeL{13}{15}{}
\EdgeL{15}{14}{}

\EdgeL{5}{9}{}
\EdgeL{6}{10}{}
\ArcR{10}{7}{}
\ArcR{3}{11}{}
\EdgeL{11}{12}{}
\ArcR{12}{8}{}

\ArcL{3}{4}{}
\ArcL{4}{8}{}

\VArcR[.4]{arcangle=65}{1}{3}{}
\VArcR[.4]{arcangle=65}{8}{14}{}
\end{VCPicture}
}
\end{center}
\caption{A graph $G$ with set of vertices $V = \{v_1,v_2, \cdots, v_{15} \}$. Note that $G$ has as a single source  $v_1$ and as a single sink $v_{14}$.}
\label{figu-graph}
\end{figure}

\begin{figure}[!t]
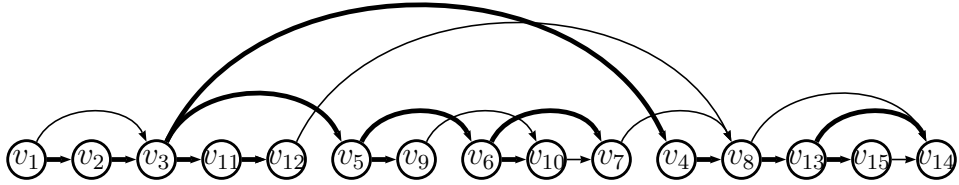

\begin{center}
\VCDraw[1.2]{
\begin{VCPicture}{(18,4)(-1,-1)}
\setlength{\unitlength}{.1in}
\State[v_1]{(0,0)}{1}
\State[v_2]{(1.2,0)}{2}
\State[v_3]{(2.4,0)}{3}
\State[v_{11}]{(3.6,0)}{11}
\State[v_{12}]{(4.8,0)}{12}
\State[v_5]{(6,0)}{5}
\State[v_9]{(7.2,0)}{9}
\State[v_6]{(8.4,0)}{6}
\State[v_{10}]{(9.6,0)}{10}
\State[v_7]{(10.8,0)}{7}
\State[v_{4}]{(12,0)}{4}
\State[v_{8}]{(13.2,0)}{8}
\State[v_{13}]{(14.4,0)}{13}
\State[v_{15}]{(15.6,0)}{15}
\State[v_{14}]{(16.8,0)}{14}
\VArcR[.4]{arcangle=65}{7}{8}{}
\VArcR[.4]{arcangle=65}{9}{10}{}
\EdgeL{11}{12}{}
\EdgeL{10}{7}{}
\EdgeL{15}{14}{}
\VArcR[.4]{arcangle=65}{8}{14}{}
\VArcR[.4]{arcangle=65}{12}{8}{}
\VArcR[.4]{arcangle=65}{1}{3}{}


\ChgEdgeLineWidth{3}
\EdgeL{1}{2}{}
\EdgeL{2}{3}{}
\VArcR[.4]{arcangle=65}{3}{4}{}
\EdgeL{4}{8}{}
\EdgeL{8}{13}{}
\VArcR[.4]{arcangle=65}{13}{14}{}
\EdgeL{13}{15}{}
\VArcR[.4]{arcangle=65}{3}{5}{}
\VArcR[.4]{arcangle=65}{5}{6}{}
\VArcR[.4]{arcangle=65}{6}{7}{}
\EdgeL{6}{10}{}
\EdgeL{5}{9}{}
\EdgeL{3}{11}{}
\EdgeL{11}{12}{}

\end{VCPicture}
}
\end{center}
\caption{Vertices of Figure~\ref{figu-graph} in topological order, where $\textit{ordD}[v_1] = 1$,  $\textit{ordD}[v_2] = 2$, $\textit{ordD}[v_3] = 3$, $\textit{ordD}[v_4] = 11$, $\textit{ordD}[v_5] = 6$, $\textit{ordD}[v_6] = 8$, $\textit{ordD}[v_7] = 10$, $\textit{ordD}[v_8] = 12$, $\textit{ordD}[v_9] = 7$, $\textit{ordD}[v_{10}] = 9$, $\textit{ordD}[v_{11}] = 4$, $\textit{ordD}[v_{12}] = 5$, $\textit{ordD}[v_{13}] = 13$, $\textit{ordD}[v_{14}] = 15$ and $\textit{ordD}[v_{15}] = 14$}
\label{figu-topological}
\end{figure}

Algorithm \Algo{SuperBubble} starts by topologically ordering the vertices of graph $G$ and then  
checks each vertex in $V$, in topological order, to identify whether it is an exit or  an entrance candidate (or both). According to Lemmas~\ref{exit} and~\ref{start}, a vertex $v$ is an exit  candidate if it has at least one parent with exactly one child (out-degree 1) and an entrance candidate if it has at least one child with exactly one parent (in-degree 1). There are at most $2n$ candidates, thus the cost of constructing a doubly-linked list of all the candidates is  linear in $n$. The elements of the candidates list are  ordered according to $\textit{ordD}$, and each candidate  is labelled as an exit or an entrance candidate. Note that if a vertex $v$ is both an exit and an entrance candidate, then $v$ appears twice in the candidates list, first as an exit and then as an entrance  (Figure~\ref{figu-cands}). 


\begin{figure}[!t]
\setlength{\tabcolsep}{5pt}
$$\begin{tabularx}{1\textwidth}{l@{}cccccccccccccccccccccccccc}
$j$ & & 1 & 2 & 3 & 4 & 5 & 6 & 7 & 8 & 9 & 10 &11&12&13&14&15 \\
\hline
~& &$v_1$ & $v_2$ & $v_3$ & $v_{11}$ & $v_{12}$ & $v_5$ & $v_9$ & $v_6$ & $v_{10}$ & $v_7$& $v_4$& $v_8$ &$v_{13}$&$v_{15}$&$v_{14}$\\

\mbox{entrance} & & $s_1$ &  & $s_2$ & $s_3$ &  & $s_4$ &  &  &  &  &  & $s_5$&$s_6$&&\\

\mbox{exit} & &  & & $t_1$ &  & $t_2$ &  &  &  & $t_3$ & $t_4$ &  & $t_5$ & & &$t_6$ \\
\end{tabularx}$$
\caption{Candidates list for Figure~\ref{figu-graph},  $\textit{candidates} = \{ v_1$(entrance),   $v_3$(exit), $v_3$(entrance), $v_{11}$(entrance), $v_{12}$(exit), $v_5$(entrance), $v_{10}$(exit), $v_7$(exit), $v_8$(exit), $v_8$(entrance), $v_{13}$(entrance), $ v_{14}$(exit)$\}.$ Note that both $v_3$ and $v_8$ appear twice in the list.}

\label{figu-cands}
\end{figure}



Algorithm \Algo{SuperBubble} processes the candidates list of graph $G$  in decreasing topological order (backwards). Let  $v'_1,v'_2,\ldots, v'_{\ell}$  be the list of candidates. The algorithm performs the following:
\begin{itemize}
\item If $v'_j$ is an entrance candidate, then delete $v'_j$;
\item If $v'_j$ is an exit candidate, then subroutine \Algo{ReportSuperBubble} is called to find  and report the  superbubble ending at $v'_j$, that is, the superbubble $\langle v'_i,v'_j \rangle$, for some entrance candidate $v'_i$. Subroutine \Algo{ReportSuperBubble} also recursively finds and reports all nested superbubbles between $v'_i$ and $v'_j$.   
\end{itemize}

For clarity of presentation, we next provide a list and a short description of subroutines and arrays used by algorithm \Algo{SuperBubble} and subroutine \Algo{ReportSuperBubble}. Before that, it is worth mentioning that $\textit{candidates}$ is a doubly-linked list of entrance and exit candidates; specifically, an element of the list is a vertex along with a label specifying if it is an entrance or an exit candidate. For the sake of simplicity of the following routines, we use a vertex and its corresponding candidate (element in the candidates list) interchangeably. This does not add to the complexity of the algorithm as we can use an auxiliary array $v$, where $v[i]$ stores a pointer to the corresponding element $v_i$ in $\textit{candidates}$ so as to provide a constant-time conversion from a vertex to the corresponding candidate.

\begin{enumerate}
\item \Algo{Entrance(\textit{v})} takes as input a vertex $v$ and outputs \texttt{TRUE} if $v$ is 
an entrance candidate, that is, if it satisfies Lemma~\ref{start}, and \texttt{FALSE} otherwise.
\item \Algo{Exit(\textit{v})} takes as input a vertex $v$ and outputs \texttt{TRUE} if $v$ is 
an exit candidate, that is, if it satisfies Lemma~\ref{exit}, and \texttt{FALSE} otherwise.
\item \Algo{InsertEntrance(\textit{v})} takes as input a vertex $v$, inserts it  at the end of \textit{candidates} and labels it as \textit{entrance}. 

\item \Algo{InsertExit(\textit{v})} takes as input a vertex $v$, inserts it  at the end of \textit{candidates} and labels it as \textit{exit}. 

\item \Algo{Head(\textit{candidates})} and \Algo{Tail(\textit{candidates})} return the first and the last element in \textit{candidates}, respectively.

\item \Algo{DeleteTail(\textit{candidates})} deletes the last element in \textit{candidates}.
\item \Algo{Next(\textit{v})} returns the candidate following $v$ in \textit{candidates}.
\end{enumerate}

In addition to the above subroutines, the main algorithm also explicitly makes use of the following arrays.  
\begin{enumerate}
\item The array $\textit{ordD}$ stores the topological order of the vertices.
\item The array $\textit{previousEntrance}$ stores the previous entrance candidate $s$ for each vertex $v$. Formally, $\textit{previousEntrance}[v]=s$ where $s$ is an entrance candidate such that $\textit{ordD}[s] < \textit{ordD}[v]$; and there does not exist another entrance candidate $s'$ such that $\textit{ordD}[s] < \textit{ordD}[s'] < \textit{ordD}[v]$. 
\item The array $\textit{alternativeEntrance}$ is used to reduce the number of $\textit{entrance}-\textit{exit}$ pairs that need to be considered as possible superbubbles. Array $\textit{alternativeEntrance}$ is further detailed in Section~\ref{subsec:validateMark}.

\end{enumerate}

%

Note that subroutine \Algo{ReportSuperBubble} is called for each exit candidate in decreasing order either by algorithm \Algo{SuperBubble} or through a recursive call to identify a nested superbubble. A call to subroutine \Algo{ReportSuperBubble}{(\textit{start, exit})} checks the possible entrance candidates between $\textit{start}$ and $\textit{exit}$, starting with the nearest previous entrance candidate (to \textit{exit}).
This task is  accomplished with the help of subroutine \Algo{ValidateSuperBubble}, explained in the following section, which checks whether a given candidate $\langle s,t \rangle$ is a superbubble or not; if it is not, the algorithm  returns either a ``-1"  which means that no superbubble ends at $t$, or an alternative entrance candidate for a superbubble that could end at $t$.
For the graph in Figure~\ref{figu-graph},  the algorithm detects and reports five superbubbles: $\langle v_8, v_{14} \rangle$, $\langle v_3, v_8\rangle$, $\langle v_5, v_7\rangle$, $\langle v_{11}, v_{12}\rangle$ and  $\langle v_1, v_3\rangle$. Here, both $\langle v_5, v_7\rangle$ and  $\langle v_{11}, v_{12}\rangle$ are nested superbubbles.

\begin{algo}{SuperBubble}{G}
\CALL{TopologicalSort}{G}
\SET{\textit{prevEnt}}{\texttt{NULL}}
\DOFOR{\mbox{each vertex $v$ in topological order}}
\SET{\textit{alternativeEntrance}[v]}{\texttt{NULL}}
\SET{\textit{previousEntrance}[\textit{v}]}{\textit{prevEnt}}
\IF{ \Call{Exit}{v}}
\CALL{InsertExit}{v}
\FI
\IF{ \Call{Entrance}{v}}
\CALL{InsertEntrance}{v}
\SET{\textit{prevEnt}}{\textit{v}}
\FI
\OD
\DOWHILE{\textit{candidates} \mbox{ is not empty}}
\IF{ \Call{Entrance}{\Call{Tail}{\textit{candidates}}}}
\CALL{DeleteTail}{\textit{candidates}}
\ELSE  \CALL{ReportSuperBubble}{\Call{Head}{\textit{candidates}},{\Call{Tail}{\textit{candidates}}}}
\FI
\OD
\end{algo}
\bigskip \bigskip

\begin{algo}{ReportSuperBubble}{\textit{start, exit}}
\COM{Report the superbubble ending at \textit{exit} (if any)}
\IF{ (\textit{start} = \texttt{NULL}) ~|~ (\textit{exit} = \texttt{NULL}) ~|~ (\textit{ordD}[\textit{start}] \geq \textit{ordD}[\textit{exit}])}
\CALL{DeleteTail}{\textit{candidates}}
\RETURN {}
\FI 

\SET{s}{\textit{previousEntrance}[{\textit{exit}}]}
\DOWHILE{ (\textit{ordD}[s] \geq \textit{ordD}[\textit{start}])   }
	\SET{\textit{valid}}{\Call{ValidateSuperBubble}{s,\textit{exit}}}
 	\IF { (\textit{valid} = s) ~|~ (\textit{valid} = \textit{alternativeEntrance}[s]) ~|~ (\textit{valid} = -1)}
		\BREAK 			
 	\FI
 	\SET{\textit{alternativeEntrance}[s]}{\textit{valid}} \label{mark}
 	\SET{s}{\textit{valid}}
 \OD
\CALL{DeleteTail}{\textit{candidates}}  	
\IF{ (\textit{valid} = s) }
   		\CALL {Report}{\langle s, \textit{exit}\rangle} \label{report}
  		\DOWHILE{(\Call{Tail}{\textit{candidates}} \mbox{ is not } s)}
  			\IF{\Call{Exit}{\Call{Tail}{\textit{candidates}}}}
   				\COM{Check for nested superbubbles}
   			 	\CALL{ReportSuperBubble}{\Call{Next}{s},\Call{Tail}{\textit{candidates}}}
   			\ELSE 
   				\CALL{DeleteTail}{\textit{candidates}}
   			\FI
   		\OD
   \FI
\RETURN{}  	
\end{algo}


\begin{rmk}
It is also possible to design the algorithm so as to move forward in topological order instead of backwards.
\end{rmk}

\noindent For graph $G$ in Figure~\ref{figu-graph}, algorithm \Algo{SuperBubble}{($G$)} makes exactly three calls to subroutine \Algo{ReportSuperBubble}:
\begin{enumerate}
\item \Algo{ReportSuperBubble}{$(v_1,v_{14})$}: First, it  checks the exit candidate $v_{14}$ against the nearest previous entrance candidate, i.e.~vertex $v_{13}$. Subroutine \Algo{ValidateSuperBubble}{$(v_{13},v_{14})$} returns $v_8$ as an alternative entrance candidate. The new candidate  is then checked and the superbubble $\langle v_{8}, v_{14} \rangle$ is reported.
\item \Algo{ReportSuperBubble}{$(v_1,v_{8})$}: First, it checks the exit candidate $v_8$ against the nearest previous entrance candidate, i.e.~vertex $v_5$.  Subroutine \Algo{ValidateSuperBubble}{$(v_5,v_8)$} returns $v_3$ as an alternative entrance candidate. The new candidate  is then checked and the superbubble $\langle v_{3}, v_{8} \rangle$ is reported. Additionally, two recursive calls are made: 
\begin{enumerate}
\item \Algo{ReportSuperBubble}{$(v_{11},v_{7})$}:  First, it validates $\langle v_{5}, v_{7} \rangle$ and reports it. Then, it makes a recursive call to subroutine \Algo{ReportSuperBubble}{$(v_{10},v_{10})$} which terminates without reporting any superbubble.
\item \Algo{ReportSuperBubble}{$(v_{11},v_{12})$}:  validates $\langle v_{11}, v_{12} \rangle$ and reports it. 
\end{enumerate}
\item \Algo{ReportSuperBubble}{$(v_1,v_3)$}: validates $\langle v_{1}, v_{3} \rangle$ and reports it.
\end{enumerate}


\section{Validating a Superbubble}\label{sec:validate}

In this section, we describe subroutine \Algo{ValidateSuperBubble}. The ability to validate a candidate superbubble depends on the following result related to the Range Minimum Query problem.

The Range Minimum Query problem, RMQ for short, is to preprocess a given array $A[1\dd n]$ for subsequent queries of the form: ``Given indices $i, j$, what is the minimum value of $A[i\dd j]$?''. The problem has been  studied intensively for decades and several $\langle O(n),O(1)\rangle$-RMQ data structures have been proposed, many of which depend on the equivalence between the Range Minimum Query and the Lowest Common Ancestor problems~\cite{harel1984,Fischer06,Durocher11}. 

In order to check whether a superbubble candidate $\langle s,t\rangle$ is a superbubble or not, we propose to utilise the range min/max  query  problem as follows: 

\begin{itemize}
\item For a given graph $G=(V,E)$ and for each vertex $v \in V$ with topological order $\textit{ordD}[v]$, calculate the topological orderings of the
parent and the child of $v$ that are topologically furthest from $v$.

\begin{align*} 
\mbox{\textit{OutParent}[\textit{ordD}$[v]]$} & =  \mbox{min (\{\textit{ordD}$[u_i] ~|~ (u_i,v) \in E\}$)}, \\
\mbox{\textit{OutChild}[\textit{ordD}$[v]]$} & =  \mbox{max (\{\textit{ordD}$[u_i] ~|~ (v,u_i) \in E\}$)}.
\end{align*}


\item For an integer array $A$ and indices $i$ and $j$ we define  $\Call{RangeMin}{A,i,j}$  and $\Call{RangeMax}{A,i,j}$ to return the  minimum and maximum value of $A[i.. j]$, respectively. 
 
 Then for a given superbubble candidate $\langle s,t\rangle$, where $s$ and $t$ are an entrance and an exit candidate respectively (satisfying Lemmas~\ref{entrance} and~\ref{exit}), if $\langle s,t\rangle$ is a superbubble then the following two conditions are valid
 
\begin{align*} 
\Algo{RangeMin}(\mbox{{\textit{OutParent},\textit{ordD}[s]+1,\textit{ordD}[t]}}) & =  \mbox{\textit{ordD}}[s], \\
\Algo{RangeMax}(\mbox{{\textit{OutChild},\textit{ordD}[s],\textit{ordD}[t]}-1}) & =  \mbox{\textit{ordD}}[t].
\end{align*} 
 
 
\end{itemize}

For example, Figure~\ref{figu-Outtable} represents both  $\textit{OutParent}$ and $\textit{OutChild}$ arrays computed for  graph $G$  in Figure~\ref{figu-graph}. Furthermore, a candidate $\langle v_5, v_8\rangle$ is not a superbubble as  $\Call{RangeMin}{\textit{OutParent},\textit{ordD}[v_5]+1,\textit{ordD}[v_8]} = 3 \neq 6 = \textit{ordD}[v_5]$.

\begin{figure}[!t]
\setlength{\tabcolsep}{4pt}
$$\begin{tabularx}{1\textwidth}{lccccccccccccccccccccccccccc}
 $j$ && 1 & 2 & 3 & 4 & 5 & 6 & 7 & 8 & 9 & 10 &11&12&13&14&15 \\
\hline
  &&$v_1$ & $v_2$ & $v_3$ & $v_{11}$ & $v_{12}$
      & $v_5$ & $v_9$ & $v_6$ & $v_{10}$ & $v_7$& $v_4$& $v_8$ &$v_{13}$&$v_{15}$&$v_{14}$\\
 ${\textit{OutParent}[j]}$ &  &-& 1 & 1 & 3 &  4 & 3 & 6 & 6 & 7 & 8 & 3&5&12&13&12\\
${\textit{OutChild}[j]}$ &  & 3& 3 & 11 & 5 & 12 & 8 & 9 & 10 & 10 & 12 & 12 & 15& 15&15&- \\
\end{tabularx}$$
\caption{\textit{OutParent} and \textit{OutChild} arrays for the graph in Figure~\ref{figu-graph}.}

\label{figu-Outtable}
\end{figure}

It should be clear that after an $\mathcal{O}(n+m)$-time preprocessing, validating a superbubble requires $\mathcal{O}(1)$ time which is the cost for range max/min query. Subroutine $\Call{ValidateSuperBubble}{\textit{startVertex, endVertex}}$ is designed to return an appropriate entrance candidate for a superbubble ending at \textit{endVertex} (if any), as follows.
\bigskip

\begin{algo}{ValidateSuperBubble}{\textit{startVertex, endVertex}}
\SET {\textit{start}}{\textit{ordD}[\textit{startVertex}]}
\SET {\textit{end}}{\textit{ordD}[\textit{endVertex}]}
\SET {\textit{outchild}}{\Call{RangeMax}{\textit{OutChild}, \textit{start}, \textit{end}-1}}
\SET {\textit{outparent}}{\Call{RangeMin}{\textit{OutParent}, \textit{start}+1, \textit{end}}}
\IF {\textit{outchild} \neq \textit{end}}
		 \RETURN{-1}
\FI
\IF {\textit{outparent} = \textit{start}}
		 \RETURN{\textit{startVertex}}
\ELSEIF { \Call{Entrance}{\Call{Vertex}{\textit{outparent}}}}  
        \RETURN{\Call{Vertex}{\textit{outparent}}}

\ELSE

		\RETURN{\textit{previousEntrance}[\Call{Vertex}{\textit{outparent}}]}
\FI  
\end{algo}
\bigskip 

Note that subroutine \Algo{ValidateSuperBubble} utilises subroutine \Algo{Entrance} and  the array \textit{previousEntrance} defined in Section~\ref{sec:find}, as well as subroutine \Algo{Vertex} that takes as input an integer $i$ and outputs vertex $v$ such that $\textit{ordD}[i] = v$.

An important observation is that a subsequent call to subroutine \Algo{ValidateSuperBubble}, for a given entrance candidate, returns alternative entrance candidates in strictly non-decreasing topological order as proved by Lemma~\ref{orderAlternative}.

\begin{lemma} \label{orderAlternative}
Let $t$ be the alternative entrance candidate returned by subroutine \Call{ValidateSuperBubble}{s,e}.
Then for any exit candidate $e'$ such that  $\textit{ordD}[s] < \textit{ordD}[e'] <  \textit{ordD}[e]$, the  order of the alternative entrance candidate $t'$ returned by subroutine \Call{ValidateSuperBubble}{s,e'} will be greater than or equal to the order of $t$.

\end{lemma}

\begin{proof}

Recall that the alternative entrance $t$  returned by the subroutine \Call{ValidateSuperBubble}{s,e'} is either a vertex with topological order \textit{outperent}, or the \textit{previousEntrance} of this vertex.

Since $\textit{outparent} = \Call{RangeMin}{\textit{OutParent}, \textit{ordD}[s]+1, \textit{ordD}[e]}$,   $\textit{outparent}' = \Call{RangeMin}{\textit{OutParent}, \textit{ordD}[s]+1, \textit{ordD}[e']}$ and $\textit{ordD}[s] < \textit{ordD}[e'] < \textit{ordD}[e]$, we have $\textit{outparent} \le \textit{outparent}'$. Therefore, $ ordD(t) \leq ordD(t')$.
\qed
\end{proof}

\subsection{Validation and \textit{alternativeEntrance}}\label{subsec:validateMark}
In case the validation of the  candidate pair $(t_0,e)$ fails, subroutine \Call{ValidateSuperBubble}{t_0,e} returns either ``-1" or an alternative candidate $t_1$ which might be an entrance of a superbubble ending at $e$. This alternative candidate $t_1$ is either a vertex $u_1$, if $u_1$ is an entrance candidate,  or the previous entrance candidate of  $u_1$  such that

\begin{align*} 
\mbox{\textit{ordD}$[u_1]$} & = \mbox{\textit{OutParent}[\textit{ordD}$[v_1]]$}\\
& = \Algo{RangeMin}(\mbox{\textit{OutParent},\textit{ordD}$[t_0]+1$,\textit{ordD}$[e]$}),
\end{align*}


\noindent where $v_1$ is some vertex between $t_0$ and $e$ in the topological ordering. 


Suppose $t_1$  is also not a valid entrance of the superbubble ending at $e$. Then there must be a vertex $v_2$, $\textit{ordD}[t_1] < \textit{ordD}[v_2] < \textit{ordD}[t_0]$, with some  parent $u_2$,  such that  $\textit{ordD}[u_2] = \textit{OutParent}[\textit{ordD}[v_2]]$. Then the alternative entrance is some $t_2$, which is either a vertex $u_2$ or its previous entrance and thus  $\textit{ordD}[t_2] < \textit{ordD}[t_1]$.  A series of such failed validations produces a sequence $t_1, t_2,...$ of failed alternative entrance candidates.


An important observation here is that any entrance $t_i$,  for $i \geq 1$,  from such a sequence   is an invalid entrance not only  for the superbubble ending at $e$  but also for all those ending at any other exit vertex $e^\prime$ such that $\textit{ordD}[t_{i-1}] < \textit{ordD}[e^\prime]< \textit{ordD}[e]$  and $t_i$ =  \Algo{ValidateSuperBubble}{($t_{i-1},e^\prime$)}. This is the case because the vertex $v_i$ which causes the alternative entrance $t_{i}$ to fail is such that $\textit{ordD}[t_{i}] < \textit{ordD}[v_i] < \textit{ordD}[t_{i-1}]$ for $i  \geq 1$. Therefore, $v_i$ does not depend on the exit $e$ but rather on  the previous failed candidate entrance.


This is where array \textit{alternativeEntrance} plays an important role. Storing $\textit{alternativeEntrance}[t_{i-1}] = t_{i}$ for $i \geq 1$ enables us to skip this sequence at a later stage if $t_i$ is returned by subroutine \Algo{ValidateSuperBubble}{$(t_{i-1},e^\prime)$}. 



\section{Algorithm Analysis}

In this section, we analyse the correctness and the running time  of the proposed algorithm \Algo{SuperBubble}.
For simplicity, in the following lemma we will slightly abuse the terminology and refer to $\langle s, t\rangle$ as a \textit{superbubble} if it satisfies the first three conditions given in Definition~\ref{def:sup}, and as \textit{minimal superbubble} if it also satisfies the last condition in the same definition.

\begin{lemma} \label{minimal}
For a given exit candidate $e$, let $s$ be 
the entrance candidate such that superbubble $\langle s,e \rangle$ is reported   by subroutine \Call{ValidateSuperBubble}{s,e}. 
Then $\langle s,e \rangle$ is a minimal superbubble.
\end{lemma}

\begin{proof}
By contradiction, let $e'$ be an exit candidate such that $\langle s,e' \rangle$ is also a superbubble and $\textit{ordD}[s] < \textit{ordD}[e'] < \textit{ordD}[e]$. Then, either $\textit{ordD}[e] = \textit{ordD}[e']+1$ or there is at least one vertex $v$ such that $\textit{ordD}[e']< \textit{ordD}[v]<\textit{ordD}[e]$.

In the first case, $\textit{ordD}[e] = \textit{ordD}[e']+1$ implies that $e$ is the only child of $e'$ and $e'$ is the only parent of $e$, which, 
by Lemma~\ref{exit} makes $\langle e',e \rangle$ a superbubble.


In the second case, where there is at least one vertex $v$ such that $\textit{ordD}[e']< \textit{ordD}[v]<\textit{ordD}[e]$,  we also argue that $\langle e', e \rangle$ must be a superbubble. Indeed, $\langle e', e \rangle$ satisfies the following conditions:
\begin{enumerate}
\item \textbf{Reachability:} Since $\langle s, e \rangle$ is a superbubble, $e$ is reachable from $s$; since $\langle s, e' \rangle$ is also assumed to be a  superbubble, any path from $s$ to $e$ must go through $e'$, therefore $e$ is reachable from $e'$.
\item\textbf{Matching:} The only vertices reachable from $e'$ without going through $e$ are those whose topological order is between  $ordD(e')$ and $ordD(e)$. Indeed, since $\langle s, e \rangle$ and $\langle s, e' \rangle$ are superbubbles, all these vertices are reachable from $s$ through $e'$, and no vertices with topological order greater than $ordD(e)$ are reachable from $e'$ without going through $e$.
Similarly,  there are no edges between vertices with topological order less than $ordD(e')$ and those with the topological order between  $ordD(e')$ and $ordD(e)$. Therefore, the only vertices from which $e$ is reachable without going through $e'$ are those whose topological order is between  $ordD(e')$ and $ordD(e)$.
\item \textbf{Acyclicity:} Since $\langle s, e \rangle$ is a superbubble it is acyclic; since $\langle e', e \rangle$ is a subgraph of $\langle s, e \rangle$, it is also acyclic.
\end{enumerate}


In both cases,  since for each exit candidate the entrance  candidates are checked in reverse topological order, subroutine \Algo{ValidateSuperBubble}  would have been called on  $\langle e', e \rangle$ first, and would have reported $\langle e', e \rangle$ instead of $\langle s, e \rangle$. Therefore, $\langle s, e \rangle$ is a minimal superbubble. \qed
\end{proof}

\begin{lemma} \label{prop:me}
For the given entrance and exit candidates $s$ and $e$, respectively, subroutine \Algo{ValidateSuperBubble} reports $\langle s, t\rangle$, if and only if, $\langle s, t\rangle$ is a superbubble. 
\end{lemma}

\begin{proof}
We prove the lemma by showing that if $\langle s, t\rangle$ is a superbubble then subroutine \Algo{ValidateSuperBubble} reports it, and if \Algo{ValidateSuperBubble} reports $\langle s, t\rangle$ then $\langle s, t\rangle$ is a superbubble.
\begin{enumerate}
\item
We start by showing that  if $\langle s, t\rangle$ is a superbubble then subroutine \Algo{ValidateSuperBubble} reports it. Indeed, by Lemma~\ref{ts}, all the vertices with topological orderings between $s$ and $t$ belong to the superbubble $\langle s, t\rangle$. Therefore, the minimum \textit{OutParent} is $s$ and the maximum \textit{OutChild} is $t$ and thus subroutine \Algo{ValidateSuperBubble} reports $\langle s, t\rangle$.
\item
We next show that if subroutine \Algo{ValidateSuperBubble} reports $\langle s, t\rangle$ then $\langle s, t\rangle$ is a superbubble. Let $\textit{start}$ and $\textit{end}$ be two integers, such that $\textit{ordD}[s] = \textit{start}$ and $\textit{ordD}[t] = \textit{end}$.
The graph $G$ as defined, has a single source $r$ and a single sink $r'$; this implies that any vertex $v\in V$ is reachable from $r$ and, at the same time, can reach $r'$. This is also true for $s$, $t$ and for any vertex $v$ such that $\textit{ordD}[s] < \textit{ordD}[v] <\textit{ordD}[t]$.  

First, we  show that $t$ is {\bf reachable} from $s$. Recall that $t$ is an exit candidate, so, it has a parent $p$ with out-degree 1. Assume that $t$ is not reachable from $s$, then there must be a path from $r \rightsquigarrow t$  which does not involve $s$. This implies that either  $\textit{OutParent}[\textit{end}] < \textit{start}$, or there exists a vertex $v$ such that $\textit{start}<\textit{ordD}[v]<\textit{end}$, $\textit{OutParent}[v] <\textit{start}$ and there exists a path $r\rightsquigarrow v \rightsquigarrow t$, which is a contradiction.

Similarly, we can show that every vertex $v$ such that   $\textit{start} < \textit{ordD}[v] <\textit{end}$ satisfies the {\bf matching} criterion of the superbubble.   

The {\bf acyclicity} criterion is guaranteed by the acyclicity of $G$ and the {\bf minimality} is satisfied by the design of subroutine \Algo{ReportSuperBubble} which assigns each exit of a superbubble to the nearest entrance, and by the correctness of Lemma~\ref{minimal}.\qed 
\end{enumerate}

\end{proof}

\begin{lemma} \label{alternativeEntrance}
For a given exit candidate $e$, let $t$ be the alternative entrance candidate returned by subroutine \Call{ValidateSuperBubble}{s,e}. Then any entrance candidate between $t$ and $e$ cannot be a valid entrance for the superbubble ending at $e$.
\end{lemma}

\begin{proof}
By contradiction, assume that $s'$ is an entrance candidate between $t$ and $e$ such that $\langle s',e \rangle$ is a superbubble. 
If $s'$ had been between $s$ and $e$, it would have already been reported, as \Algo{SuperBubble} checks entrance candidates  in reverse topological order starting from $e$.  
Therefore, $s'$ is between $t$ and $s$, such that $\textit{ordD}[t] < \textit{ordD}[s'] < \textit{ordD}[s] < \textit{ordD}[e]$. Let $\textit{outparent} = \Call{RangeMin}{\textit{OutParent}, \textit{ordD}[s]+1, \textit{ordD}[e]}$. Then, vertex at \textit{outparent} is between $t$ and $s'$, otherwise subroutine \Call{ValidateSuperBubble}{s,e} would have returned $s'$ (instead of $t$). Therefore, $\textit{ordD}[t] \leq \textit{outparent} < \textit{ordD}[s']$.

Let $\textit{outparent}' = \Call{RangeMin}{\textit{OutParent}, \textit{ordD}[s']+1, \textit{ordD}[e]}$. Then $\textit{outparent}' \leq \textit{outparent}$. This implies that 
$\textit{outparent}' \leq \textit{outparent} < \textit{ordD}[s']$. However, for $\langle s', e \rangle$ to be a valid superbubble, $\textit{outparent}'$ should have been equal to $\textit{ordD}[s']$. Hence, the assumption is wrong and thus, it is  proved that there cannot be an entrance candidate, between $t$ and $e$, which is a valid entrance for the superbubble ending at $e$. \qed
\end{proof}

\begin{lemma} \label{resetMark}
For the given   entrance and exit candidates $s$ and $e_1$, respectively, let $\textit{alternativeEntrance}[s]$ be set to $t_1$ which later gets reset to $t_2$ such that $t_2 \ne t_1$, while considering $s$ with another exit candidate $e_2$. Then no entrance candidate between $s$ and $e_2$ can reset $\textit{alternativeEntrance}[s]$ to $t_1$ again.
\end{lemma}

\begin{proof}
Let $e_3$  be an exit candidate between $s$ and $e_2$ such that subroutine \Call{ValidateSuperBubble}{s,e_3} returns $t_3$. Then by Lemma~\ref{orderAlternative}, $\textit{ordD}[t_1] \leq  \textit{ordD}[t_2] \leq \textit{ordD}[t_3]$. Since $t_1 \ne t_2$, we have 
$\textit{ordD}[t_1] <  \textit{ordD}[t_2] \leq \textit{ordD}[t_3]$.
Therefore, $t_1 < t_3$ and $\textit{alternativeEntrance}[s]$ cannot be reset to the same value $t_1$ again. \qed
\end{proof}

\begin{thm}\label{thm2}
Algorithm \Algo{SuperBubble} reports all superbubbles, and only superbubbles,  in graph $G$ in decreasing topological order of their exit vertices in $\mathcal{O}(n+m)$-time.
\end{thm} 
\begin{proof}
Consider an execution  of algorithm \Algo{SuperBubble}. Let superbubbles $\langle s_1,t_1\rangle,\cdots,\langle s_k,t_k\rangle $ be the successive superbubbles reported  just after the execution of Line~\ref{report} of subroutine \Algo{ReportSuperBubble}, where $\textit{ordD}(t_1) > \textit{ordD}(t_2) >\cdots>\textit{ordD}(t_k)$.

\begin{enumerate}
\item First, we show that each $\langle s_i,t_j\rangle$  reported by the algorithm in Line~\ref{report} is  a superbubble. This is proved by the correctness of Lemma~\ref{prop:me}. 

\item Second, no superbubble is missed out by the algorithm as proved by the following.
Subroutine \Algo{ReportSuperBubble} is called for each exit candidate in decreasing order. The entrance candidate for the superbubble (if any) ending at \textit{exit} will only be between $\textit{start}$ and $\textit{exit}$, where $\textit{start}$ is either the head of the the candidates list (when subroutine \Algo{ReportSuperBubble} is called from algorithm \Algo{SuperBubble}) or next candidate of the entrance of an outer superbubble (when called through a recursive call to identify a nested superbubble). A call to subroutine \Algo{ReportSuperBubble}{(\textit{start, exit})} checks the possible entrance candidates between $\textit{start}$ and $\textit{exit}$, starting with the nearest previous entrance candidate (to \textit{exit}). Subroutine \Algo{ValidateSuperBubble} either successfully validates an entrance candidate, or returns a ``-1", or returns an alternative entrance candidate. From Lemma~\ref{alternativeEntrance}, there cannot be any valid entrance between this alternative entrance and \textit{exit}. If this alternative entrance starts a sequence of entrances already checked for some exit candidate previously (as depicted by \textit{alternativeEntrance}), then all entrances of that sequence will be skipped, otherwise this alternative entrance will be tested. However, as mentioned in Section~\ref{subsec:validateMark}, none of the entrance candidates in the skipped sequence can be valid. Therefore, for each exit candidate, every potential entrance candidate is checked for validity, and those which are not considered are not valid. 

\item Third, the running time  of \Algo{SuperBubble} is $\mathcal{O}(n+m)$. Indeed, the running time of the \Algo{TopologicalSort} and computing the candidates list is $\mathcal{O}(n+m)$.
Furthermore, all list operations cost constant time each, and sum up to a linear cost of $\mathcal{O}(n)$, as there are at most $2n$ candidates in the list.  Finally, each call for subroutine \Algo{ValidateSuperBubble} costs $\mathcal{O}(1)$. The total number of times \Algo{ValidateSuperBubble}  is called is $\mathcal{O}(n+m)$.  This is because subroutine \Algo{ValidateSuperBubble} is called once for each exit candidate in Line 7 of subroutine \Algo{ReportSuperBubble}, and the total number of such calls is bounded by $\mathcal{O}(n)$. Additionally,  it is called every time a new \textit{alternativeEntrance} sequence is generated by subroutine \Algo{ValidateSuperBubble}.
It follows from Lemma~\ref{resetMark} that once an \textit{alternativeEntrance} sequence is reset, it cannot be generated again by subsequent calls to subroutine \Algo{ValidateSuperBubble}. This resetting of \textit{alternativeEntrance} for each entrance candidate  (Line~\ref{mark}) thus enables avoiding repeated checks of the same sequences of entrance candidates. Resetting is done every time an edge is considered for the first time between a vertex (in between an entrance candidate \textit{startVertex} and an exit candidate \textit{endVertex}) and its topologically furthest parent (whose order is less than that of \textit{startVertex}). Thus, the total number of times \textit{alternativeEntrance} will be reset (for all the entrance 
candidates) is bounded by $\mathcal{O}(m)$.

Therefore, the total running time for reporting all superbubbles in graph $G$ is $\mathcal{O}(n+m)$. 
\end{enumerate}
\qed 
\end{proof}

\section{Final Remarks}
We presented an $\mathcal{O}(n+m)$-time algorithm to compute all superbubbles in a directed acyclic graph, where $n$ is the number of vertices and $m$ is the number of edges, improving the best-known $\mathcal{O}(m \log m)$-time algorithm for this problem~\cite{SungSSBP14}. 
It is also interesting to note that in this type of graph, that is, constructed from sequences over a fixed-sized alphabet, the out-degree of each vertex is bounded by the size of the alphabet (four for DNA alphabet); therefore, the time complexity of the proposed algorithm is essentially linear in $n$.

Our immediate goal is to practically evaluate our algorithm and compare its implementation to an earlier result~\cite{OnoderaSS13}. It would also be interesting to investigate other superbubble-like structures in assembly graphs, such as complex bulges~\cite{bulges}.    


\bibliographystyle{elsarticle-num}
\bibliography{paperbib}

\end{document}